\documentclass{article}
\usepackage{style/variables}
\Shortfalse
\commentsfalse

\usepackage{authblk}
\usepackage{microtype}
\usepackage{style/packages}
\usepackage{style/comments}
\usepackage{style/shortlong}
\usepackage{style/macros}

\usepackage{hyperref}
\usepackage[margin=1in]{geometry}

\usepackage{style/environments} % theorems, definitions, etc.

\title{Rectilinear Link Diameter and Radius in a Rectilinear Polygonal Domain\footnote{An extended abstract appeared at the 29th International Symposium on Algorithms and Computation (ISAAC 2018)~\cite{isaacversion}.
     EA was supported by the SNF Early Postdoc Mobility grant P2TIP2-168563, Switzerland, F.R.S.-FNRS, Belgium, and by the Foundation for the Advancement of Theoretical Physics and Mathematics ``BASIS'', Russia.
     MC, AvR and MR were supported by JST ERATO Grant Number JPMJER1201, Japan.
     MC was also supported in part by ERC StG 757609.
     MK was supported in part by KAKENHI No. 17K12635, Japan and NSF award CCF-1422311.
     AM was supported by the Netherlands' Organisation for Scientific Research (NWO) under project no.~024.002.003.
     YO was partially supported by JSPS KAKENHI Grant Number 15K00009 and JST CREST Grant Number JPMJCR1402, and Kayamori Foundation of Informational Science Advancement.
     AO was supported by the Fund for Research Training in Industry
and Agriculture (FRIA).}}

\author[1]{Elena Arseneva}
\author[2]{Man-Kwun Chiu}
\author[3]{Matias Korman}
\author[4]{Aleksandar Markovic}
\author[5,6]{Yoshio Okamoto}
\author[7]{Aur\'elien Ooms}
\author[8]{Andr\'{e} van Renssen}
\author[4]{Marcel Roeloffzen}
\affil[1]{St. Petersburg State University, St. Petersburg, Russia\\
  \texttt{e.arseneva@spbu.ru}}
\affil[2]{Institut f\"ur Informatik, Freie Universit\"at Berlin, Berlin, Germany\\
  \texttt{chiumk@zedat.fu-berlin.de}}
\affil[3]{Tufts University, Boston, USA\\
  \texttt{matias.korman@tufts.edu}}
\affil[4]{TU Eindhoven, Eindhoven, the Netherlands\\
  \texttt{\{a.markovic,m.j.m.roeloffzen\}@tue.nl}}
\affil[5]{University of Electro-Communications, Tokyo, Japan\\
  \texttt{okamotoy@uec.ac.jp}}
\affil[6]{RIKEN Center for Advanced Intelligent Project, Tokyo, Japan}
\affil[7]{Universit\'e libre de Bruxelles (ULB), Brussels, Belgium\\ 
  \texttt{aureooms@ulb.ac.be}}
\affil[8]{University of Sydney, Sydney, Australia\\ 
  \texttt{andre.vanrenssen@sydney.edu.au}}

\date{}

\begin{document}

\maketitle

%!TEX root = ../conference.tex
% Makes possible to compile from that file in Texmaker
\begin{abstract}
We study the computation of the diameter and radius under the rectilinear link
distance within a rectilinear polygonal domain of $n$ vertices and $h$ holes.
We introduce a \emph{graph of oriented distances} to encode the distance between pairs of points of the domain. This helps us transform the problem so that we can search through the candidates more efficiently. 
Our algorithm computes both the diameter and the radius in
$O (\min(n^\omega, n^2 + nh \log h + \chi^2))$ time, where
$\omega<2.373$ denotes the matrix multiplication exponent and $\chi\in
\Omega(n)\cap O(n^2)$ is the number of edges of the graph of
oriented distances. We also provide an alternative algorithm for computing the diameter that runs in $O(n^2 \log n)$ time.

%the radius in ical and horizontal segments we show how:
%\begin{enumerate}
%\item Compute APSP in $O(n^2 \log n)$ time with $O(\log{n})$ query time,
%\item Compute its diameter in $O(n^2 \log n)$ time,
%\item Compute its radius in $\min \{\,O(n^\omega), O(n^2 \log n + \chi^2)\,\}$ time,
%\end{enumerate}
%where $\chi$ denotes the number of intersecting rectangles in the horizontal and vertical decomposition of the polygonal domain.
\end{abstract}

%!TEX root = ../arxiv.tex
% Makes possible to compile from that file in Texmaker

\section{Introduction}

Diameters and radii are popular characteristics of metric spaces. For a compact set $S$ with a metric $d\colon S\times S \to \mathbb{R}^+$, its diameter is defined as
$
\mathrm{diam}(S) := \max_{p \in S}\max_{q \in S} d(p,q),
$ 
and its radius is defined as
$
\mathrm{rad}(S) := \min_{p \in S}\max_{q \in S} d(p,q).
$
The pair $(p,q)$ and the point $p$ that realize these distances
are called the \emph{diametral pair} and \emph{center}, respectively.
These terms are the natural extension of the same concepts in a disk and give some interesting properties of the environment, such as the worst-case response time
or ideal location of a serving facility.

Much research has been devoted towards finding efficient algorithms to compute the diameter and radius for various types of sets and metrics. In computational geometry, one of the most
well-studied and natural metric spaces is a polygon in the plane. This paper focuses on the computation of the diameter and the radius of a rectilinear polygon, possibly with holes
(i.e., a \emph{rectilinear polygonal domain}) under
the \emph{rectilinear link distance}. Intuitively, this metric measures the minimum number of links (segments) required in any rectilinear path connecting two points in the domain, where rectilinear indicates that we are restricted to horizontal and vertical segments only. 

\subsection{Previous Work}
The ordinary link distance is a very natural metric and simple to describe. Initially, the interest was motivated by the potential robotics applications (i.e., having some kind of robot
with wheels for which moving in a straight line is easy, but making turns is costly in time or energy). Since then, it has attracted a lot of attention from a theoretical point of view.

Indeed, many problems that are easy under the $L_1$ or Euclidean metric turn out to be more challenging under the link distance. For example,
the shortest path between two points in a polygonal domain of combinatorial complexity $n$ can be found in $O(n\log n)$ time for both Euclidean~\cite{DBLP:journals/siamcomp/HershbergerS99}
and $L_1$ metrics~\cite{ref:MitchellAn89,ref:MitchellL192}. However, even approximating the shortest path within a factor of $(2-\epsilon)$ under the link distance is 3-SUM hard~\cite{DBLP:journals/comgeo/MitchellPS14},
and thus it is unlikely that a significantly subquadratic-time algorithm is possible.
Recently the bit complexity of this problem was studied~\cite{DBLP:journals/jocg/KostitsynaLPS17} and it was shown that sometimes
$\Omega(n\log{n})$ bits is required to represent coordinates of some vertices of the optimal path.

The problem of computing the diameter and radius is no exception to this rule: when polygons are simple (i.e., they do not have holes) and have $n$ vertices, the diameter and center can be found in linear time for both Euclidean~\cite{DBLP:journals/dcg/AhnBBCKO16,DBLP:journals/siamcomp/HershbergerS97} and $L_1$ metrics~\cite{DBLP:journals/comgeo/BaeKOW15}. However, the best known algorithm for the link distance runs in $O(n\log n)$ time~\cite{DBLP:journals/dcg/DjidjevLS92,DBLP:journals/jcss/Suri89}. Lowering the running times or proving the impossibility of this is a longstanding open problem in the field. The only partial answer to this question was given by Nilsson and Schuierer~\cite{DBLP:conf/cga/NilssonS91,DBLP:journals/comgeo/NilssonS96}; they showed that the diameter and center can be found in linear time under the rectilinear link distance (i.e., when we are only allowed to use rectilinear paths).

We focus on polygons with holes. The addition of holes to the domain introduces significant difficulties to the problem. For example, the diameter and radius under the rectilinear link distance can be uniquely realized by points in the interior of a polygonal domain (see Figure~\ref{fig:diam-not-vertices}). Hence, it does not suffice to determine the distance only between every pair of vertices of the domain. Other strange situations can happen, such as the diameter and radius being arbitrarily close (see e.g. Figure~\ref{fig:diam_radius}).

These difficulties have a clear impact in the runtime of the algorithms. In most metrics, the runtime changes from linear or slightly superlinear to large polynomial terms. The difference between the link distance and other metrics becomes even more significant: no algorithm for computing the diameter and radius under the link distance is known, not even one that runs in exponential time (or one that works for particular cases such as rectilinear polygons). A summary of the best running time for computing the diameter and center under different metrics can be found in Table~\ref{table:summary}.

In this paper we provide the first step towards understanding such a difficult metric. Similarly to the simple polygon case~\cite{DBLP:conf/cga/NilssonS91,DBLP:journals/comgeo/NilssonS96}, we start by considering the computation of both the diameter and radius under the rectilinear link distance. We hope that the ideas of this paper will motivate future research in solving the more difficult problem of computing the diameter and radius under the (ordinary) link distance.

\begin{table}[t]
\centering
\caption{Summary of the best known results for computing the diameter and radius of a polygonal domain of $n$ vertices and $h$ holes under
different metrics. In the table,
$\omega<2.373$ is the matrix multiplication exponent.
}%
\label{table:summary}
\begin{tabular}{|c|cc||cc|}
\cline{1-5}
Metric &
\multicolumn{2}{c||}{Simple polygon} &
\multicolumn{2}{c|}{Polygonal domain}
\\
\hline
 & Diameter & Radius & Diameter & Radius \\
\cline{2-5}
Euclidean &
$O(n)$~\cite{DBLP:journals/siamcomp/HershbergerS97} &
$O(n)$~\cite{DBLP:journals/dcg/AhnBBCKO16} &
$O(n^{7.73})$~\cite{DBLP:journals/dcg/BaeKO13} &
$O(n^{11}\log n)$~\cite{DBLP:journals/jocg/Wang18} \\
$L_1$ &
$O(n)$~\cite{DBLP:journals/comgeo/BaeKOW15} &
$O(n)$~\cite{DBLP:journals/comgeo/BaeKOW15} &
$O(n^2+h^4)$~\cite{DBLP:journals/dcg/BaeKMOPW17} &
$\tilde{O}(n^4+n^2h^4)$~\cite{DBLP:journals/dcg/BaeKMOPW17} \\
Ordinary link &
$O(n\log n)$~\cite{DBLP:journals/jcss/Suri89} &
$O(n\log n)$~\cite{DBLP:journals/dcg/DjidjevLS92} &
open &
open
\\
Rectilinear link &
$O(n)$~\cite{DBLP:conf/cga/NilssonS91} &
$O(n)$~\cite{DBLP:journals/comgeo/NilssonS96} &
$O(n^2\log n)$~(Thm.~\ref{theo:DiameterAlgorithm}) &
$O(n^\omega)$~(Thm.~\ref{theo:MatrixMultiplication}) 
\\
\cline{1-5}
\end{tabular}
\end{table}

\begin{figure}[!tbp]
  \centering
  \begin{minipage}[t]{0.48\textwidth}
  \centering
    \includegraphics[scale=0.8]{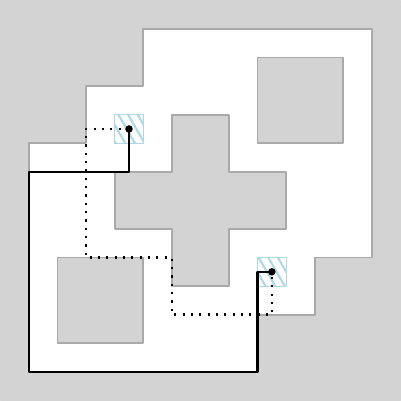}
    \caption{An example showing no diametral pair
lies on the boundary of the polygonal domain. The points in the dashed blue regions will have distance 6
from each other (out of the 4 shortest paths connecting them two are shown) whereas other pairs will have distance 5 or less.}
    \label{fig:diam-not-vertices}
  \end{minipage}
  \hfill
    \begin{minipage}[t]{0.48\textwidth}
    \centering
    \includegraphics[scale=0.8]{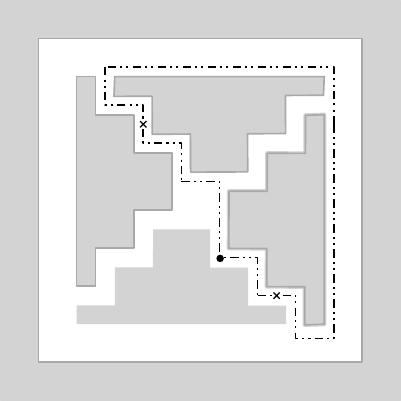}
    \caption{Example with diameter 8 (crossed points) and radius 7 (dotted point).
    By increasing the number of bends in the holes the diameter and radius become arbitrarily close. Note that any point in the domain is either a center or belongs to a diametral pair.}\label{fig:diam_radius}
  \end{minipage}
\end{figure}
\subsection{Results}
Several of the difficulties of the link distance disappear when restricting the problem to a rectilinear setting. For example, one can easily
partition the domain into $O(n^2)$ rectangular cells such that all points in a cell have the same distance to all points in another cell (for some domains the number of such cells is $\Theta(n^2)$). With this partition, brute-force algorithms that find the diameter and radius in $O(n^3\log n)$ time immediately follow. Alternatively, you can use a slightly coarser method to approximate either value: in $O(n^2 + nh \log h)$
or $O(n^2\log\log n)$ time we can compute an estimate of the diameter (details of these methods are given in Section~\ref{sec_orient}). This estimate will either be the exact diameter
or will be the diameter plus one (i.e., the path computed may contain an additional link that is not needed).

In our work we improve this second approach. By using some geometric observations, we characterize exactly when the estimate is off by one unit.
Thus, we can transform the approximation algorithm into an exact one by adding a verification step that checks whether or not the one additive error has actually happened.

We provide three different algorithms for making the above additional verification step. In Section~\ref{sec:Characterization} we characterize what we should look for to determine what the exact diameter is. This characterization then leads to a brute-force algorithm that runs in $O(n^2 + nh \log h+ \chi^2)$ time, where $\chi$ is a parameter of the input that ranges from $\Theta(n)$ to $\Theta(n^2)$. To reduce running times when $\chi$ is large we present another algorithm to compute the diameter in Section~\ref{sec:DiameterImprovement}. This algorithm, which runs in $O(n^2 \log n)$ time, exploits properties of the diameter. Specifically, we heavily use that this value is a maximum over a maximum of distances, hence it can only be used for the diameter (recall that we have a minimum-maximum alternation in the definition of the radius). For the radius we then present a third algorithm that uses matrix multiplication to speed up computation. This solution runs in time $O(n^{\omega})$, where $\omega < 2.373$ is the matrix multiplication exponent (Le Gall~\cite{DBLP:conf/issac/Gall14a} provided
the best known bound on~$\omega$). This last solution can also be adapted to compute the diameter, but our second algorithm results in a faster method.

Another interesting benefit of our approach is that we may be able to obtain a {\em certificate}. In previous algorithms for computing the diameter or center in polygonal domains, the diameter is found via exhaustive search. Thus, even if somehow the points that realize the diameter or center are given, the only way to verify that the answer is correct is to run the whole algorithm. In our algorithm, knowing the diameter can reduce the time needed for verification. Although the reduction in computation time is not large (from $O(n^2\log n)$ for computing to $O(n^2\log \log n)$ for verifying the diameter, for example), we find it to be of theoretical interest.

Further note that, when comparing with the algorithms for other metrics, 
the running time for simple and polygonal domains  differs by at least a cubic factor. In our case, running times only
increase by a slightly superlinear factor when compared to the case of simple polygons~\cite{DBLP:conf/cga/NilssonS91,DBLP:journals/comgeo/NilssonS96}. This is partially due
to the fact that rectilinear link distance is much easier than the ordinary link distance, but also because we use this new verification approach. We believe
this to be our main contribution and hope that it motivates a similar approach in other metrics.

\subsection{Preliminaries}

A \emph{rectilinear simple polygon} (also called an orthogonal polygon)
is a simple polygon that has horizontal and vertical edges only. A
\emph{rectilinear polygonal domain} $P$ with $h$ pairwise disjoint
holes and $n$ vertices is a connected and compact subset of
$\mathbb{R}^2$ with $h$ pairwise disjoint holes, in which
the boundary of each hole is a simple closed rectilinear curve.
Thus, the boundary $\partial P$ of $P$ consists of $n$ line segments.

Each of the holes as well as the outer boundary of $P$ is regarded as an \emph{obstacle} that  paths in $P$ are not allowed to cross. A \emph{rectilinear path~$\pi$} from $p \in P $ to $q \in P $ is a path from~$p$ to~$q$ that consists of vertical and horizontal segments, each contained in~$P$, and such that along~$\pi$ each vertical segment is followed by a horizontal one and vice versa. Recall that $P$ is a closed set, so $\pi$ can traverse the boundary of $P$ (along the outer face and any of the $h$ obstacles).

We define the \emph{link length} of such a path to be the number of segments composing it. The \emph{rectilinear link distance} between points~$p,q\in P$ is defined as the minimum link length of a rectilinear path from $p$ to $q$ in $P$,
and denoted by $\rld(p,q)$.
 It is well known that in rectilinear polygonal domains there always exists a rectilinear polygonal path between any two points $p, q \in P$, and thus the distance is well defined.
Once the
distance is defined, the definitions of \emph{rectilinear link diameter} $\mathrm{diam}(P)$ and \emph{rectilinear link radius} $\mathrm{rad}(P)$ directly follow.

For simplicity in the description, we assume that a pair of vertices do not share the same $x$- or $y$-coordinate unless they are connected by an edge. This general position assumption can be removed with classic symbolic perturbation techniques. Also notice that, since we are considering rectilinear polygons, no edge has length~0. However, for simplicity in the analysis we will allow edges in a rectilinear path to have length~0. These edges of length 0 are considered as edges and thus potentially contribute to the link distance (naturally, no shortest path will ever have such an edge). The reason for considering these is that we will consider oriented paths, where the first and last edge are forced to be horizontal or vertical, this enforcement may require edges of length 0. From now on, for ease of reading, we will refer to rectilinear simple polygons and rectilinear polygonal domains as ``simple polygons'' and ``domains.'' Similarly, we will use the term ``distance'' to refer to the rectilinear link distance.

\section{Graph of Oriented Distances}\label{sec_orient}

In this section we introduce the graph of oriented distances and show how it can be used to encode the rectilinear link distance between points of the domain. We note that, although we have not been able to find a reference to this graph in the literature, some properties are already known. For example, the horizontal and vertical decompositions (defined below) were used by Mitchell {\em et al.}~\cite{mpsw-oamlrptrd-19} to compute minimum-link rectilinear paths.

For any domain $P$, we
extend any horizontal segment of the domain to the left and right until it hits another segment of $P$,
partitioning it into rectangles. We call this partition the \emph{horizontal decomposition} (see Figure~\ref{fig:hor-ver-decomp}, left). Let \( \hdec(P) \) be the set containing those rectangles. Similarly, if we
extend all the vertical segments up and down,
we get the \emph{vertical decomposition} (see Figure~\ref{fig:hor-ver-decomp}, right). Let  \( \vdec(P) \) be the set of rectangles in this second decomposition. Observe that both decompositions have linear size and can be computed in $O(n\log n)$ time with a  plane sweep.

The overlay of both subdivisions creates a finer subdivision that has the well-known property that pairwise cell distance is constant (that is, the distance between any pair of points in two fixed cells of this subdivision will remain constant). Thus, by computing the distance between all pairs of cells we can find both the diameter and center. The major problem of this approach is that the finer subdivision may have $\Omega(n^2)$ cells, and thus it is hard to obtain an algorithm that runs in subcubic time. Instead, we avoid the overlay and use both subdivisions separately to obtain the diameter.

\shortlong{}{
\begin{figure}
\centering
\includegraphics{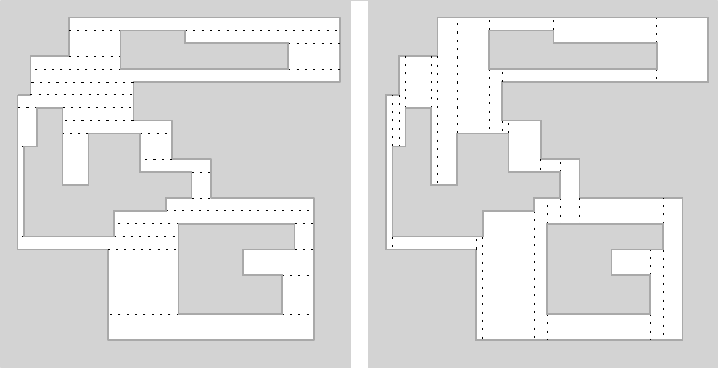}
\caption{The horizontal and vertical decompositions of the domain. Note that, because of our general position assumption, both subdivisions have the same number of rectangles.}
\label{fig:hor-ver-decomp}
\end{figure}
} %shortlong

Given two rectangles $i,j \in \hdec(P) \cup \vdec(P)$, we use $i \sqcap j$ to denote the boolean operation which returns \emph{true} if and only if
the rectangles $i$ and $j$ properly intersect (i.e. their intersection has non-zero area). This implies that one of $i,j$ belongs to $\hdec(P)$, and the other to $\vdec(P)$.

\begin{definition}[Graph of Oriented Distances]
Given a rectilinear polygonal domain $P$, let $\orgr(P)$ be the unweighted undirected graph defined as
$\orgr(P) = (\hdec(P) \cup \vdec(P),\{\, (h,v) \in \hdec(P) \times \vdec(P) \colon\, h \sqcap v 
\,\})$.
\iffalse
\begin{displaymath}
\orgr(P) = (
	\hdec(P) \cup \vdec(P),
	\{\, (h,v) \in \hdec(P) \times \vdec(P) \colon\, h \sqcap v 
\,\}
).
\end{displaymath}
\fi
\end{definition}

In other words, vertices of  \( \orgr(P) \) correspond to rectangles of the horizontal and the vertical decompositions of \( P \). We add an edge between two vertices if and only if the corresponding rectangles properly intersect. Note that this graph is bipartite, and has $O(n)$ vertices. From now on, we make a slight abuse of notation and identify a rectangle with its corresponding vertex (thus, we talk about the neighbors of a rectangle $i \in \hdec(P)$ in $\orgr(P)$, for example).

The name \emph{Graph of Oriented Distances} is 
explained as follows (see also the paragraph after Lemma~\ref{lem:lpq}). Consider a rectilinear path $\pi$ between two points in $P$. Each horizontal edge of $\pi$ is contained in a rectangle of $\hdec(P)$ and each vertical edge is contained in a rectangle of $\vdec(P)$. A bend in the path takes place in the intersection of the rectangles containing the two adjacent edges and corresponds to an edge of $\orgr(P)$. So every rectilinear path $\pi$ has a corresponding walk $\pi'$ in $\orgr(P)$ (and {\em vice versa}). Moreover, each bend of $\pi$ is associated with an edge of $\pi'$. 

\begin{definition}[Oriented distance]
	Given a rectilinear polygonal domain \( P \), let \( i \) and \( j \)
	be two vertices of \( \orgr(P) \), let \( \ordist(i,j)\) to be the length of
	the shortest path between $i$ and $j$ in graph \( \orgr(P) \) plus one. We also define \( \ordist(i,i) = 1\).
\end{definition}
The reason why we add the extra unit is to make sure that the link distance and the oriented distance match (see Lemma~\ref{lem:lpq} below). We first list some useful properties of the oriented distance, which directly follow from the definition.
Then we show the relationship between the oriented distance $\ordist(\cdot,\cdot)$ in $\orgr(P)$
and the link distance $\rld(\cdot,\cdot)$ in $P$.

\begin{lemma}\label{lemma:oriented-dist-props}
Let \( i, j, i', j' \) be any (not necessarily distinct) rectangles in $\hdec(P)\cup \vdec(P)$
such that \( i \sqcap i'\), and $j \sqcap j'$. Then, the following hold.
\begin{enumerate}[leftmargin=*,label=(\alph*),ref=(\alph*)]
\item $\ordist(i,j) = \ordist(j,i).$
\item\label{lem:flip-i}
$ \ordist(i',j) \in \{\,
    \ordist(i,j) - 1,
    \ordist(i,j) + 1
  \,\}.
$

\item\label{lem:flip-ij}
$
  \ordist(i',j') \in \{\,
    \ordist(i,j) - 2,
    \ordist(i,j),
    \ordist(i,j) + 2
  \,\}.
$
\end{enumerate}
\end{lemma}
\shortlong{}{
\begin{proof}
The first statement follows from the fact that any path can be followed in reverse order. Thus, a path from $i$ to $j$ immediately transforms into a path of the same length from $j$ to $i$. The second statement follows from the fact that $i \sqcap i'$, hence $i$ and $i'$ are adjacent in $\orgr(P)$. By following that edge we have that the distance from any $j$ to $i$ and $i'$ can differ by at most one unit. The two values cannot be equal by the fact that $i \in \hdec(P)\Leftrightarrow i'\in \vdec(P)$ and since they are adjacent they are in different sets of the bipartite graph, and hence, their distance to any other node cannot be the same.
The third statement follows from the second statement applied to the two pairs $(i,i')$ and $(j,j')$.
\end{proof}
} %shortlong

\begin{lemma}\label{lem:lpq}
Let \( p \) and \( q \) be two points of the rectilinear polygonal domain \( P \).
The rectilinear link distance $\rld(p,q)$ between $p$ and $q$ can be characterized as follows.
If $p$ and $q$ lie in the same vertical or horizontal rectangle of $\vdec(P)$ or $\hdec(P)$ then $\rld(p,q) = 1$ (if $p$ and $q$ share a coordinate) or $\rld(p,q) = 2$ (if both $x$- and $y$-coordinates of $p$ and $q$ are distinct). Otherwise,
let \( i \in \hdec(P) \), \( i' \in \vdec(P) \), \( j \in \hdec(P) \) and \( j' \in
\vdec(P) \) be vertices of the graph of oriented distances such that \( p \in i \cap i' \) and \( q \in j \cap j' \). 
Then 
\begin{displaymath}
	\rld(p,q) = \min \{\,
		\ordist(i,j),
		\ordist(i,j'),
		\ordist(i',j),
		\ordist(i',j')
	\,\}.
\end{displaymath}
\end{lemma}
\shortlong{%Proof omitted in conference version
}{
\begin{proof}
The case in which $p$ and $q$ lie in the same
rectangle $R$ is easy: in either case we can connect them with either a single segment or a path with exactly one bend and stay within $R$. Since we are within $R$ the path is feasible and has the minimum number of links possible.

Now suppose that there is no rectangle that contains both $p$ and $q$. Let $\pi$ be a
shortest rectilinear path from $p$ to $q$ in $P$. It cannot lie entirely in one
rectangle. Assume that the first and the last link of $\pi$ are horizontal. Then
the length of $\pi$ is equal to $\ordist(i,j)$. Moreover, if one of
$\ordist(i,j'), \ordist(i',j), \ordist(i',j')$ was strictly
smaller than $\ordist(i,j)$, then it would correspond to a path $\pi'$
shorter than $\pi$. Analogous arguments for other  orientations
of the initial and final links of $\pi$ complete the proof.
\end{proof}
} %End long

Intuitively speaking, if we are given two disjoint rectangles $i,j \in \hdec(P)$, then $\ordist(i,j)$ denotes the minimum number of links needed to connect any two points $p\in i$ and $q\in j$ under the constraint that
the first and the last segments of the path are horizontal.
If we looked for rectangles in $\vdec(P)$, we would instead require that the path starts (or ends) with vertical segments.
It follows that the link distance is the minimum of the four possible options.

In our algorithms we will often look for oriented distances between rectangles, so we compute it and store them in a preprocessing phase. Fortunately, a similar decomposition was used by Mitchell~\etal~\cite{mpsw-oamlrptrd-19}. Specifically, they show how to compute the distance from a single rectangle to all other rectangles in $O(n + h \log h)$ time with an $O(n)$-size data structure.\footnote{As a subproblem towards obtaining their main result, Mitchell~\etal~\cite{mpsw-oamlrptrd-19} show how to compute the distance from a single point to any other location in the domain with paths of fixed orientation. They call these the $h$-$h$-map, $v$-$v$-map, $v$-$h$-map and $h$-$v$-map and they correspond to our rectangular decompositions. Although their method considers a single starting point, it can be adapted to compute the distance from a rectangle as all points inside each rectangle we consider will have the same resulting distances to the other rectangles.}

\begin{lemma}
\label{lem:OrientedDistanceComputation}{\cite{mpsw-oamlrptrd-19}}
Given the horizontal and vertical decompositions $\hdec(P)$ and $\vdec(P)$ we can compute for a single rectangle $i$ in either decomposition the oriented distance $\ordist(i,j)$ to every other rectangle $j$ in $O(n + h \log h)$ time.
\end{lemma}
We construct this data structure for each of the $O(n)$ rectangles. This allows us to compute (and store) the $O(n^2)$ oriented distances in $O(n^2 + nh \log h)$ time. Alternatively, we can use a recent result by Chan and Skrepetos~\cite{cs-apspgig-19} to compute the same distances in $O(n^2\log\log n)$ time.
\matias{Added the new result by Chan. Note that it does not always superceed Mitchell et al.}

%!TEX root = ../conference.tex
% Makes possible to compile from that file in Texmaker

\section{Characterization via Boolean Formulas}
\label{sec:Characterization}

Let $\ordiam = \max_{i,j \in \hdec(P) \cup \vdec(P)}  \ordist(i,j)$ be the largest distance between vertices of $\orgr(P)$. Similarly, we define $\orrad = \min_{i \in \hdec(P) \cup \vdec(P)} \max_{j \in \hdec(P) \cup \vdec(P)}  \ordist(i,j)$. Note that these two values are the diameter and the radius of \( \orgr(P) \) plus one (recall that we add one unit to the graph distance when defining $\ordist$). We use $\ordiam$ and $\orrad$ to approximate the diameter $\mathrm{diam}(P)$ and radius $\mathrm{rad}(P)$ of a domain $P$ under the rectilinear link distance.  First, we relate the distance between two points $p,q \in P$ to the oriented distances between the rectangles that contain $p$ and $q$. Specifically, from Lemma~\ref{lem:lpq}, we know that $\rld(p,q) = \min \{\,\ordist(i,j),\ordist(i,j'),\ordist(i',j),\ordist(i',j')	\,\}$, where $i,j \in \hdec(P)$ are the horizontal rectangles containing $p$ and $q$, respectively, and $i',j' \in \vdec(P)$ are the vertical rectangles containing $p$ and $q$. Similarly, we define $\orell(p,q) = \max \{\,\ordist(i,j),\ordist(i,j'),\ordist(i',j),\ordist(i',j')\,\}$. It then follows from Lemma~\ref{lemma:oriented-dist-props} that these two values differ by at most~2.

\begin{lemma}\label{lem:ellbound}
For any two points $p,q \in P$, let $i,j \in \hdec(P)$ and $i',j' \in \vdec(P)$ be the rectangles containing $p$ and $q$, i.e., $p \in i\cap i'$ and $q \in j \cap j'$. 
Then, it holds that $\orell(p,q) - 2 \leq \rld(p,q) \leq \orell(p,q) - 1$.
\end{lemma}

This relation allows us to express the rectilinear link diameter of a domain in terms of $\ordiam$.

\begin{theorem}\label{the:diameter-formula}
  The rectilinear link diameter $\rlddiam(P)$ of a rectilinear polygonal domain $P$ satisfies $\rlddiam(P) = \ordiam-1$ if and only if there exist
$i,i',j,j' \in \hdec(P) \cup \vdec(P)$ with $i\sqcap i'$ and $j \sqcap j'$, such that 
 $ \ordist( i, j) = \ordiam \text{ and } \ordist(i',j') = \ordiam$.
   Otherwise, $\mathrm{diam}(P) = \ordiam-2$.
\end{theorem}
\begin{proof}
Before giving our proof, we emphasize that the fact that $\mathrm{diam}(P) \in\{\ordiam-1,\ordiam -2\}$ is %a folklore fact
folklore (although we have found no reference, several researchers mentioned %the result were aware of it
that they were aware of it). Our major contribution is the characterization of which of the two cases it is.

Now observe that for any pair of points $p,q \in P$ we have $\rld(p,q) \leq \orell(p,q) - 1 \leq \ordiam - 1$ by Lemma~\ref{lem:ellbound}. Hence, the diameter of $P$ 
is at most $\ordiam - 1$. Similarly, by the definitions of $\ordiam$ and $\orell(\cdot,\cdot)$, there must be a pair of points $p,q \in P$ so that $\orell(p,q) = \ordiam$. Again by Lemma~\ref{lem:ellbound} it follows that $\rlddiam(P) \geq \rld(p,q) \geq \orell(p,q) - 2 = \ordiam - 2$.
  
  Next we show that the diameter is $\ordiam - 1$ if and only if the above condition holds. If $\ordist(i,j) = \ordiam$ and $\ordist(i',j') = \ordiam$, then by Lemma~\ref{lemma:oriented-dist-props} and the fact that neither $\ordist(i,j')$ nor $\ordist(i',j)$ can be larger than $\ordiam$, we know that $\ordist(i,j') = \ordist(i',j) = \ordiam - 1$. It follows from Lemma~\ref{lem:lpq} that a pair of points $p \in i \cap i'$ and $q \in j \cap j'$ has $\rld(p,q) = \ordiam - 1$. Thus, the diameter is $\ordiam -1$. 
  
  Now consider any pair $p,q$ and the set of rectangles $i,j \in \hdec(P)$ and $i',j' \in \vdec(P)$ with $p \in i\cap i'$ and $q \in j \cap j'$. Recall that $\rld(p,q) = \min\{\ordist(i,j), \ordist(i,j'), \ordist(j',i), \ordist(i',j')\}$. By Lemma~\ref{lemma:oriented-dist-props}, $\ordist(i,j)$ and $\ordist(i',j')$ must differ by exactly one from $\ordist(i',j)$ and $\ordist(i,j')$. That implies that two distances may be $\ordiam -1$, but if the condition in the lemma is not satisfied, at most one can be $\ordiam$ and the fourth must be $\ordiam - 2$ or less. Therefore, if the condition is not satisfied for $i,i',j,j'$, then the diameter is indeed $\ordiam - 2$.
\end{proof}

For the radius we can make a similar argument. 
\begin{theorem}\label{the:radius-formula}
  The rectilinear link radius $\rldrad(P)$ of a rectilinear polygonal domain $P$ satisfies 
$\rldrad(P) = \orrad - 1$ if and only if for all 
$i,i' \in \hdec(P) \cup \vdec(P)$ with $i\sqcap i'$
there exist $j,j' \in \hdec(P) \cup \vdec(P)$ with $j\sqcap j'$ such that 
$  \ordist( i, j) \geq \orrad \text{ and } \ordist(i',j') \geq \orrad $.
  Otherwise, $\rldrad(P) = \orrad - 2$.
\end{theorem}

\begin{proof}
We first show by contradiction that the real radius satisfies $\rldrad(P) \leq \orrad - 1$. Suppose the radius is greater than or equal to $\orrad$. Then, for all $p \in P$ there exists a point $q \in P$ such that $\rld(p,q) \geq \orrad$. Now consider a rectangle $i \in  \hdec(P)\cup \vdec(P)$, a point $p \in i$ and a point $q$ at distance $\orrad$ from $p$. Consider the two rectangles $j \in \hdec(P)$ and $j' \in \vdec(P)$ so that $q \in j \cap j'$.
By Lemma~\ref{lem:lpq} we know that $\ordist(i,j) \geq \rld(p,q) \geq \orrad$ and $\ordist(i,j') \geq \rld(p,q) \geq \orrad$. By Lemma~\ref{lemma:oriented-dist-props}b $\ordist(i,j)$ and $\ordist(i,j')$ differ by one,  and thus one of them must be at least $\orrad + 1$. That is, for any rectangle $i$ we can find a second rectangle at oriented distance $\orrad +1$. This implies that $\orrad = \min_{i \in \hdec(P) \cup \vdec(P)} \max_{j \in \hdec(P) \cup \vdec(P)}  \ordist(i,j) \geq \orrad + 1$, which is a contradiction. Therefore, our initial assumption that $\rldrad(P) \geq \orrad$ is false and we conclude that $\rldrad(P) \leq \orrad - 1$. 

Next we show that $\rldrad(P) \geq \orrad - 2$. Consider any point $p$ and a rectangle $i \in \hdec(P)$ that contains it. By definition of $\orrad$ there is a rectangle $j \in \hdec(P)\cup \vdec(P)$ so that $\ordist(i,j) \geq \orrad$. Let $q$ be any point in $j$. From Lemma~\ref{lem:ellbound} we get that $\rld(p,q) \geq \orell(p,q) - 2 \geq \ordist(i,j) - 2 \geq \orrad - 2$. Hence for any point $p$, there is a point $q$ that is at distance at least $\orrad - 2$, which implies $\rldrad(P) \geq \orrad - 2$.

Now we show that if the above condition is satisfied, then it must hold that $\rldrad(P) = \orrad - 1$. 
Assume the condition holds and consider any point $p$ and 
two rectangles $i,i' \in  \hdec(P) \cup \vdec(P)$ so that $i \sqcap i'$ and $p \in i\cap i'$.
There exist $j,j' \in \hdec(P) \cup \vdec(P)$ so that $j \sqcap j'$, $\ordist( i, j) \geq \orrad$, and $\ordist(i',j') \geq \orrad$.
By Lemma~\ref{lemma:oriented-dist-props} we know that $\ordist(i,j')$ and $\ordist(i',j)$ must be at least $\orrad - 1$. Therefore $\rld(p,q) \geq\orrad - 1$ for any point $q \in j\cap j'$. This shows that for any point $p$ there is a point $q$ whose link distance to $p$ is at least $\orrad - 1$, giving a lower bound on the radius. Combining this with the upper bound shown above, we obtain that $\rldrad(P) = \orrad - 1$ as claimed.

If the condition is not true, then we know there exist rectangles 
$i,i' \in \hdec(P) \cup \vdec(P)$ so that $i \sqcap i'$, and for every 
$j,j' \in \hdec(P) \cup \vdec(P)$ with $j \sqcap j'$ the above statement is not true. 
Now consider a point $p \in i\cap i'$. We argue that $p$ has distance at most $\orrad - 2$ to any other point $q\in P$. 
Consider any point~$q$ and let $j,j' \in \hdec(P) \cup \vdec(P)$ be the rectangles containing~$q$. 
We perform a case analysis on the value of $\ordist(i,j)$. First consider the case $\ordist(i,j) \geq \orrad +1$. In this case $\ordist(i',j) \geq \orrad$ and $\ordist(i,j') \geq \orrad$ which contradicts our assumption that the above statement is not true for every $(j,j')$.
If $\ordist(i,j) = \orrad$, then by Lemma~\ref{lemma:oriented-dist-props} and the assumption that not both $\ordist(i,j) \geq \orrad$ and $\ordist(i',j') \geq \orrad$ we find that $\ordist(i',j') = \orrad - 2$ which implies that $\rld(p,q) \leq \orrad - 2$.
If $\ordist(i,j) = \orrad - 1$, then by Lemma~\ref{lemma:oriented-dist-props}, both $\ordist(i,j')$ and $\ordist(i',j)$ differ from $\ordist(i,j)$ by 1, but by our assumption that not both $\ordist(i,j') \geq \orrad$ and $\ordist(i',j) \geq \orrad$, one of them must be $\orrad - 2$.
Lastly, if $\ordist(i,j) \leq \orrad - 2$, we can already conclude that $\rld(p,q) \leq \orrad - 2$. This shows that from $p$ any other point $q$ is at most distance $\orrad - 2$ away, hence the radius is at most $\orrad - 2$. Combining this with the lower bound of $\orrad -2$ (shown above), we conclude that the radius must be $\orrad - 2$.
\end{proof}

With the above characterization, we can naively compute the diameter and the radius by checking all \( O(n^4) \) quadruples \( (i,i',j,j') \in \hdec(P) \times \vdec(P)  \times \hdec(P) \times \vdec(P) \). However, the approach can be improved by using $\orgr(P)$.

\begin{corollary}%
\label{theo:ChiAlgorithm}
The rectilinear link diameter $\rlddiam(P)$ and radius $\rldrad(P)$ of a rectilinear polygonal domain $P$
consisting of $n$ vertices and $h$ holes can be computed in $O(n^2 + nh \log h + \chi^2)$ time,
where $\chi$ is the number of edges of $\orgr(P)$ (i.e., the number of pairs of intersecting rectangles of $\HDec(P)$ and
$\VDec(P)$).
\end{corollary}
\shortlong{%No proof in conference version
}{
\begin{proof}
First, we use Lemma~\ref{lem:OrientedDistanceComputation} to compute the oriented distance $\OrDist(\cdot,\cdot)$ between each pair of oriented rectangles from the horizontal and the vertical decompositions in $O(n^2 + nh \log h)$ time. 

Recall that $\orgr(P)$ adds an edge between two rectangles $i, i'$ if and only if $i \sqcap i'$. Thus, rather than looking at all quadruples $(i, i', j, j') \in\hdec(P) \times \vdec(P)  \times \hdec(P) \times \vdec(P)$, we can look at pairs of edges of $\orgr(P)$. For each of the $O(\chi^2)$ pairs we can test the conditions from Theorems~\ref{the:diameter-formula}~and~\ref{the:radius-formula} in a brute-force manner in constant time which gives us the claimed running time.
\end{proof}
} %end long

As we discuss later, this method is only useful when $\chi$ is very small, i.e. almost linear size or smaller. 

\paragraph*{Remark on the interior realization of the diameter/radius}
Theorems~\ref{the:diameter-formula}~and~\ref{the:radius-formula} together with Lemma~\ref{lemma:oriented-dist-props}b imply 
that a necessary condition for the diameter to be uniquely realized by pairs of points in the interior of $P$ is that $\rlddiam(P) = \ordiam-1$. 
Similarly, for all centers to be determined by points in the interior we must have $\rldrad(P) = \orrad-1$. 
However, neither condition is sufficient. 
This transformation of the problem into a search of quadruples of rectangles allows us to handle the 
\emph{interior cases} (when all centers or, respectively, all members of all diameter pairs are interior points)
 in the same way as the \emph{boundary cases} (when at least 
one of the centers or, respectively, at least one point of some diameter pair lies on the boundary of $P$).

% !TEX root = ../conference.tex

\section{Computing the Diameter Faster}%
\label{sec:DiameterImprovement}
We present a faster method for computing the diameter. This method uses the fact that the diameter is defined as a maximum over maxima which allows us to reduce the running time to $O(n^2\log n)$. Recall that the radius is a minimum over maxima, thus the algorithm of this section does not trivially extend to the computation of the radius. The rest of this section is the proof of the following statement.

\begin{theorem}\label{theo:DiameterAlgorithm}

  The rectilinear link diameter $\rlddiam(P)$ of a rectilinear polygonal domain $P$ of $n$ vertices can be computed in $O(n^2 \log n)$ time.

\end{theorem}

By Theorem~\ref{the:diameter-formula}, after we compute the
oriented diameter $\ordiam$, we only need to consider $\ordiam-1$ or
$\ordiam-2$ as candidates to be $\rlddiam(P)$. The following corollary of %
Theorem~\ref{the:diameter-formula} can be obtained by 
applying Lemma~\ref{lemma:oriented-dist-props}c.
\begin{corollary}\label{cor:diameter-21}

  The diameter $\rlddiam(P)$ equals~$\ordiam-2$ if and only if
  for all rectangles~$i$ and~$j$ with~$\ordist( i, j) = \ordiam$,
  and for all rectangles~$i'$ and~$j'$ with~$i \sqcap i'$
  and~$j \sqcap j'$, we have~$\ordist(i',j') = \ordiam - 2$.
  Otherwise, $\rlddiam(P) = \ordiam-1$.

\end{corollary}
This condition can be checked in \(O(n^4)\) time in a brute-force manner as follows.
We iterate over every pair $(i,j)$ with $\ordist(i,j) = \ordiam$. For each such pair we find the sets $\cover(i) = \{i' : i\sqcap i'\}$ and $\cover(j) = \{j' : j \sqcap j'\}$. Then for each pair $(i',j') \in \cover(i) \times \cover(j)$ we check if $\ordist(i',j') = \ordiam - 2$. If there is a pair for which this is not the case, then by the above corollary the diameter is $\ordiam - 1$. Since each of the covers may have linear size, the running time is $\Omega(n^4)$.

The key observation that allows us to reduce this to $O(n^2 \log n)$ time is that in the end there are only $O(n^2)$ unique pairs to test. Indeed, what we are checking is the distance of every pair $(i',j')$ in the set
\begin{displaymath}
\mathcal{T} = \{ (i',j') \colon\, \exists i,j \text{ such that } (i' \sqcap i, j \sqcap j', \ordist(i,j) = \ordiam) \}
\end{displaymath}
which clearly has only quadratic size. Next we show that this set has more structure than just being an arbitrary set of rectangles, which allows us to compute it more quickly.

First, instead of iterating over every pair $(i,j)$ with $\ordist(i,j) = \ordiam$ and computing all pairs in $\cover(i) \times \cover(j)$, we iterate over $i$ and compute all pairs in $\cover(i) \times \bigcup_{j \colon \ordist(i,j) = \ordiam} \cover(j)$. For a rectangle $i \in \hdec(P) \cup \vdec(P)$, let $\mathcal{S}_i$ denote the set of rectangles at oriented distance $\ordiam$ from $i$. Now let
\begin{displaymath}
 \mathcal{T} = \bigcup_i \mathcal{T}_i = \bigcup_i \{(i',j') \colon\, \exists j \text{ such that } (i' \sqcap i , j' \sqcap j, j \in \mathcal{S}_i) \}.
\end{displaymath}

Note that the rectangles fulfilling the role of $i'$ are easily found (i.e., they must intersect $i$ and must have different orientation), but naively computing the ones that fulfill the role of $j'$ leads to a quadratic runtime. That is, if we were to compute for each $j\in \mathcal{S}_i$ its cover, then this may take $\Omega(n^2)$ time. However, there are only $O(n)$ rectangles that can fulfill the role of $j'$ and we show how to find them in $O(n\log n)$ time.

For this purpose we use an orthogonal segment
intersection reporting data structure, derived from a
known dynamic vertical ray shooting data structure~\cite{gk-rayshooting-2009}.
The data structure we use stores horizontal line segments. It allows to add or
remove horizontal line segments in \(O(\log n)\) time per segment.
The structure reports the first segment hit by a vertical query ray in $O(\log n)$ time. By repeatedly using the structure, we can find all \(z\) horizontal line segments intersected by a
vertical line segment in \(O((z+1) \log n)\) time.
While performing the query, we also remove all the reported segments from the data structure in the same time complexity.

For a rectangle $k$, we define the
\emph{middle segment} $\ell_k$ of $k$. If $k$ is a horizontal rectangle, i.e., $k \in \hdec(P)$, then $\ell_k$ is the line
segment connecting the midpoints of its left and right boundary; if $k$ is a vertical
rectangle, i.e., $k \in \vdec(P)$, then $\ell_k$ is the segment connecting the midpoints of its top and bottom boundary.

We fix a rectangle \(i\), and assume without loss of generality that the rectangles in $\mathcal{S}_i$ are vertical.
Insert the middle segments of all horizontal rectangles in $\hdec(P)$  into the
intersection reporting data structure. Then, for each rectangle $j \in
\mathcal{S}_i$, we query its corresponding middle segment.
By the definition of middle segments, each reported
horizontal segment corresponds to a rectangle $j'$ intersecting $j$. Since we remove each segment as we find it, no rectangle is reported twice.
Repeating this for all $j \in \mathcal{S}_i$
finds the set $\mathcal{C}_i = \{\, j' \colon\, j' \sqcap j , j \in \mathcal{S}_i\,\}$ of
  all horizontal rectangles that intersect at least one rectangle in $\mathcal{S}_i$.
Each query can be charged either to the horizontal segment that
is deleted from the data structure or, in case \(z = 0\),  to the rectangle $j \in \mathcal{S}_i$ that we
are querying. Hence, the total query time sums to $O(n \log n)$.

 For each rectangle in the set $\mathcal{C}_i$, we
should check the distance to every rectangle \(i'\) such that \(i' \sqcap i\).
Doing this explicitly takes $O(n^2)$ time. Thus, summing over all rectangles
\(i\), we get the total running time of \(O(n^3)\).

To bring the running time down to $O(n^2 \log n)$, we create a reverse map of the map
\( i  \mapsto \mathcal{C}_i\). For each rectangle $k$, we build a collection
$\mathcal{L}_k$ that contains $i$ if and only if \(k\) belongs to
$\mathcal{C}_i$. Given a rectangle $j'$, we need to check the distance between
$j'$ and $i'$ for any $(i,i')$ with $i \in \mathcal{L}_{j'}$ and $i \sqcap i'$.
Using the intersection reporting data structure, we
compute for each rectangle $j'$ the set $\mathcal{D}_{j'}$, which is the set of all rectangles intersecting those in $\mathcal{L}_{j'}$.
For each rectangle $i' \in \mathcal{D}_{j'}$, we test
if $\ordist(i',j') = \ordiam - 2$. 
Again recall that  
if we find a pair with $\ordiam$,
then the diameter must be $\ordiam-1$ (otherwise, the diameter is $\ordiam-2$).
This proves Theorem~\ref{theo:DiameterAlgorithm}.

%!TEX root = ../conference.tex
% Makes possible to compile from that file in Texmaker

\section{Computation via Matrix Multiplication}
\label{sec:MatrixMultiplication}

In this section we provide an alternative method to compute the diameter and radius. This method also uses the conditions in Theorem~\ref{the:diameter-formula} and \ref{the:radius-formula}, but instead exploits the behavior of matrix multiplication on (0,1)-matrices. Recall that, given two (0,1)-matrices $A$ and $B$, their product is ${(AB)}_{i,j} = \sum_k (A_{i,k} \cdot B_{k,j}) = | \{\, k \colon\, A_{i,k} = 1 \land B_{k,j} = 1\,\}|$.

We define a (0,1)-matrix \(I\), which is used to compute both the diameter and radius:
\begin{displaymath}
I_{i,j} = \begin{cases}
1 & \text{ if } i \sqcap j, \\
0 & \text{ otherwise.}
\end{cases}
\end{displaymath}
In other words, for each pair $i,j$ of rectangles in
$\hdec(P) \cup \vdec(P)$, the matrix $I$ indicates
whether $i$ and $j$ intersect and have different orientations
(one horizontal, one vertical). Note that, for ease of explanation, we have slightly abused the notation and identified rectangles of $\hdec(P) \cup \vdec(P)$ with indices in the matrix.

\subsection{Computing the Diameter} 
We use Theorem~\ref{the:diameter-formula} to compute the diameter. Thus, we need to determine if there exist four rectangles in $\hdec(P) \cup \vdec(P)$ that satisfy the condition of Theorem~\ref{the:diameter-formula}. If so, the diameter will be $\ordiam-1$; otherwise, $\ordiam -2$.
In order to do so, we  define the (0,1)-matrix $D$ that
indicates, for a pair $i,j$ of rectangles in $\hdec(P) \cup \vdec(P)$, whether the
oriented distance between them is $\ordiam$:
\begin{displaymath}
D_{i,j} = \begin{cases}
1 & \text{ if } \ordist(i,j) = \ordiam, \\
0 & \text{ otherwise.}
\end{cases}
\end{displaymath}

By multiplying \(I\) and \(D\), we obtain
\begin{displaymath}
{(ID)}_{i,j'} =
| \{\, i' \colon\, (i \sqcap i') \land (\ordist(i',j') = \ordiam)\,\} |.
\end{displaymath}

In other words, the entry at $(i,j')$ of the product $ID$ counts the number of
rectangles in $\hdec(P) \cup \vdec(P)$ that intersect rectangle $i$ and are oriented
differently from it, and at the same time are at oriented distance $\ordiam$
from rectangle $j'$.
We construct the (0,1)-matrix \(M\) that records when an entry in the product
$ID$ is non-zero:
\begin{displaymath}
M_{i,j} = \begin{cases}
1 & \text{ if } {(ID)}_{i,j} > 0, \\
0 & \text{ otherwise.}
\end{cases}
\end{displaymath}

Finally, we look at the product \(DM\). Note that \((DM)_{i,i'} > 0\) if
and only if
there are two rectangles~$j$ and~$j'$ with~$j \sqcap j'$ such 
that~$\ordist(i,j) = \ordiam$ and~$\ordist(i',j') = \ordiam $

The quantifier on \(j'\) and the condition on its intersection with \(j\) can
be moved just to the right of the quantifier on \(j\) without altering the
meaning of the formula, since both of them are existential quantifiers. 
Therefore, the condition in Theorem~\ref{the:diameter-formula} is satisfied if and only if there exists a $1$-entry in \(I\) whose corresponding entry in \(DM\) is non-zero. This condition can be checked in quadratic time (once matrix $DM$ has been computed) by iterating over the entries of the matrices in quadratic time since the matrices have linearly many rows and columns. 

\subsection{Computing the Radius}
A similar construction can be used to verify the condition in Theorem~\ref{the:radius-formula} and compute the radius.
Similar to the matrix \(D\) given above, we define the (0,1)-matrix \(R\)
which indicates whether a pair of rectangles is at oriented distance at least
$\OrRad$ from each other:
\begin{displaymath}
R_{i,j} = \begin{cases}
1 & \text{ if } \ordist(i,j) \geq \OrRad, \\
0 & \text{ otherwise.}
\end{cases}
\end{displaymath}

By multiplying \(I\) and \(R\), we obtain
\begin{displaymath}
	{(IR)}_{i,j'} =
	| \{\, i' \colon\, (i \sqcap i') \land (\ordist(i',j') \geq \OrRad)\,\} |.
\end{displaymath}

Analogous to matrix \(M\), we 
that indicates whether the corresponding entry of $IR$ is non-zero, as follows:
\begin{displaymath}
N_{i,j} = \begin{cases}
1 & \text{ if } {(IR)}_{i,j} > 0, \\
0 & \text{ otherwise.}
\end{cases}
\end{displaymath}

We now look at the product \(RN\). Note that \((RN)_{i,i'} > 0\) if
and only if
there are two rectangles~$j$ and~$j'$ with~$j \sqcap j'$ such 
that~$\ordist(i,j) \geq \orrad$ and~$\ordist(i',j') \geq \orrad $

By a similar argument as in the diameter case, the condition on Theorem~\ref{the:radius-formula}
 is satisfied if and only if for each $1$-entry in
\(I\) the corresponding entry in \(RN\) is non-zero.
As before, this condition can be checked by iterating over the entries of the matrices in quadratic time once the matrix $RN$ has been computed.

Note that the time taken by the computation of the various matrix products
dominates the time taken by the other loops and operations. 
Each matrix has $O(n)$ rows and columns,
and the product of two \(O(n) \times O(n)\) matrices can
be computed in \( O(n^\omega) \) time. 
We summarize the results of this section in the following theorem. 
\begin{theorem}%
\label{theo:MatrixMultiplication}
  The rectilinear link diameter $\rlddiam(P)$ and radius $\rldrad(P)$ of a rectilinear polygonal domain
  $P$ consisting of $n$ vertices can be computed in \(O(n^\omega)\) time.
\end{theorem}

% !TEX root = ../conference.tex
\section{Conclusions}

Our algorithms heavily rely on Theorems~\ref{the:diameter-formula} and~\ref{the:radius-formula}. They implicitly do a search among a list of candidates for the diametral pair or radius,
but as mentioned in the introduction they have the advantage that can be used to yield a certificate. For example, if the diameter is $\ordiam-1$, it suffices to report the four rectangles
that satisfy the condition of Theorem~\ref{the:diameter-formula}. Then, if someone wants to verify that any pair in the intersection of the two pairs forms diameter, they have to compute
the oriented distances (to obtain $\ordiam$) and then verify that indeed the four given rectangles satisfy the property of Theorem~\ref{the:diameter-formula} (a similar certificate can be
obtained when the radius is $\orrad-1$ and Theorem~\ref{the:radius-formula}).

The verification runtime is dominated by the time needed to compute the oriented distance between all pairs of rectangles (which can be done in $O(n^2\log \log n)$ time using the result of Chan and Skrepetos, and is not known whether it can be done in subquadratic time~\cite{cs-apspgig-19}). We wonder if a similar certificate approach can be designed for the case in which the diameter is $\ordiam-2$, the radius is $\orrad-2$ and/or other metrics, specifically for the classic link distance (for which no algorithm is yet known).

We note that $\ordiam$ and $\orrad$ can be used to give an approximation of the diameter and radius, respectively, with an additive error of only one unit. However, the running time of computing these two values is almost as large as computing the exact values. It would be interesting to see if there is another way to approximate the diameter or radius.

\shortlong{ 
\begin{wrapfigure}[13]{r}{4cm}
\centering
\vspace{-1em}
\includegraphics{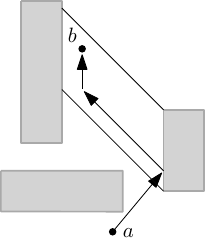}
\caption{Illustration of why our results do not extend to $8$ orientations. }
\label{fig:c-oriented}
\end{wrapfigure}
}{}

This consideration, together with our results, remind us of recent
lower bound results in fine-grained complexity.  For $n$-vertex sparse
unweighted undirected graphs, under the orthogonal vectors conjecture,
computing the diameter and even approximating it within a factor of
$3/2-\varepsilon$ cannot be done in $O(n^{2-o(1)})$ time for any
$\varepsilon > 0$ \cite{DBLP:conf/stoc/RodittyW13},
and
under the hitting set conjecture, computing the radius and even
approximating it within a factor of $3/2-\varepsilon$ cannot be done in
$O(n^{2-o(1)})$ time for any $\varepsilon > 0$
\cite{DBLP:conf/soda/AbboudWW16}.
It is known that both the strong exponential-time hypothesis
and the hitting set conjecture individually
imply the orthogonal vectors conjecture (see Vassilevska Williams' survey \cite{DBLP:conf/iwpec/Williams15}).
As we already pointed out, the (not rectilinear) link distance
computation is 3-SUM hard \cite{DBLP:journals/comgeo/MitchellPS14},
and it is straightforward to adapt the proof to show that the (not
rectilinear) link diameter computation is 3-SUM hard, too.
However, we have been unable to show the 3-SUM hardness of computing
the rectilinear link diameter or radius, nor any hardness of having
an $O(n^{2-o(1)})$-time algorithm based on the
strong exponential-time hypothesis, the orthogonal vectors conjecture,
or the hitting set conjecture so far.
Such a result would show that our algorithms are close to optimal.

\shortlong{ }{
\begin{wrapfigure}[15]{r}{4cm}
\centering
\vspace{-1em}
\includegraphics{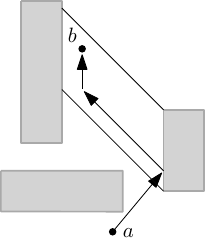}
\caption{Illustration of why our results do not extend to $8$ orientations.}
\label{fig:c-oriented}
\end{wrapfigure}
}

A natural way to extend our results would be to consider the $c$-oriented link distance. In this distance we are allowed to use $c$ slopes in our path,
normally in steps of $2\pi/c$ radians. In general allowing more directions would create more natural paths and potentially much shorter paths. Unfortunately
non-orthogonal directions create some problems. The duality between the graph $\orgr(P)$ and oriented distance relies on the observation that within a horizontal
rectangle we can reach every point with just one bend, regardless of where we enter the rectangle with a vertical ray. When the directions are no longer orthogonal
we cannot give this guarantee anymore (see Figure~\ref{fig:c-oriented}). This may be solved by counting bends that ``skip'' over orientations to count heavier.
That is, within the path, we force adjacent edges to have adjacent orientations, even if those edges have length 0. For example, in the 8-oriented distance, an L-shaped
path would have a cost of 3. However, it is not clear how to extend our results to this model. More importantly, this distance would differ a lot from the classical link distance.

\bibliography{bibliography/conference}

\end{document}